\documentclass[conference]{IEEEtran}

\usepackage{times}
\usepackage[hidelinks]{hyperref}
\usepackage{cite}
\usepackage{amsmath,amssymb,amsthm}
\usepackage{mathtools}
\usepackage{makecell}
\usepackage{algorithm}
\usepackage[noend]{algpseudocode}

\hypersetup{pdfstartview=FitH,pdfpagelayout=SinglePage}

\usepackage{tikz}
\usepackage{graphicx}
\usepackage[capitalize]{cleveref}
\usepackage{comment}
\usepackage{flushend}
\usepackage{siunitx}
\usepackage{enumitem}
\usepackage{xifthen}
\usepackage[breakable]{tcolorbox}
\usepackage{subcaption}

\usetikzlibrary{spline}
\usetikzlibrary{calc}
\tikzset{
    asnode/.style={
        circle, 
        draw=black!60, 
        shading=radial,
        outer color={rgb,255:red,137;green,207;blue,240},
        inner color=white, 
        thick, 
        minimum size=8mm,
    }
}

\newcommand{\pan}{PAN}
\renewcommand{\d}[1]{\,\mathrm{d}#1}
\newtheorem{theorem}{Theorem}

\newcommand{\markus}[1]{\bgroup\color{purple}Markus: #1\egroup}

\crefname{section}{\S}{Sections}
\crefname{page}{page}{pages}
\crefname{paragraph}{\S}{Sections}
\Crefname{section}{Section}{Sections}
\crefformat{section}{\S#2#1#3}
\Crefformat{section}{Section~#2#1#3}

\SetEnumitemKey{props}{%
label={\sffamily P\arabic*},
align=left,
ref=P\arabic*}

\definecolor{lightgray}{gray}{0.95}
{%
    \begin{tcolorbox}[breakable, colback=lightgray, colframe=lightgray, left=1pt,right=1pt,top=0pt,bottom=0pt]%
    \itshape
    \ifthenelse{\isempty{#1}}{}{\noindent\textbf{#1}\quad}%
}%
{\end{tcolorbox}}

\begin{document}

\title{\texorpdfstring{Enabling Novel Interconnection Agreements\\with Path-Aware Networking Architectures}{Enabling Novel Interconnection Agreements\\with Path-Aware Networking Architectures}}

\author{Simon Scherrer\\Department of\\Computer Science\\ETH Zurich \and Markus Legner\\Department of\\Computer Science\\ETH Zurich \and Adrian Perrig\\Department of\\Computer Science\\ETH Zurich \and Stefan Schmid\\Faculty of\\Computer Science\\University of Vienna}

\maketitle
\thispagestyle{plain} 
\pagestyle{plain}

\begin{abstract}
Path-aware networks (PANs) are emerging as an intriguing new paradigm
with the potential to significantly improve the
dependability and efficiency of networks. 
However, the benefits of PANs 
can only be realized if the adoption of such architectures
 is economically viable. 
This paper shows that PANs
enable novel interconnection agreements among autonomous
systems, which allow to considerably improve
both economic profits and
path diversity compared to today's Internet.
Specifically, by supporting
packet forwarding along a path selected by
the packet source, PANs do not 
require the Gao--Rex\-ford conditions to ensure 
stability. Hence,   
autonomous systems can establish novel agreements,
creating new paths
which demonstrably improve latency and bandwidth metrics 
in many cases.
This paper also expounds two methods to
set up agreements
which are Pareto-optimal, fair, and thus attractive
to both parties. We further present a bargaining mechanism that
allows two parties to efficiently automate agreement negotiations. 
\end{abstract}

\section{Introduction}
\label{sec:introduction}

\emph{Path-aware networks} (\pan{s}) 
are an innovative networking paradigm
which has the potential to improve the dependability and efficiency
of networks by increasing the flexibility
in packet forwarding. In contrast to today's Internet, \pan{s}
enable end-hosts to choose the path at the level of 
autonomous systems (ASes), which is then embedded in the header of data packets.
As such, they are not limited to using a single
path between a pair of ASes, but can use all available paths 
simultaneously. This multi-path approach has
two important consequences.
First, the availability of multiple paths increases 
the network's  resilience to link failures and its overall
capacity through the possibility to avoid congested links.
Second, the possibility of path selection enables end-hosts to
choose paths based on their applications' requirements---e.g.,
low latency for voice-over-IP calls and high bandwidth for
file transfers.

Over the past two decades, significant progress has
been made regarding the technical
prerequisites for realizing these 
\pan{} benefits. Numerous PAN architectures have been proposed---%
including
Platypus~\cite{Raghavan2004,Raghavan2009}, 
PoMo~\cite{Bhattacharjee2006}, NIRA~\cite{Yang2007},
Pathlets~\cite{godfrey2009pathlet}, NEBULA~\cite{Anderson2013},
and SCION~\cite{Zhang2011,perrig2017scion}---all
of which enable end-host
path selection between
provider-acknowledged paths
(in contrast to source
routing, where end-hosts
are trusted to
construct paths themselves).
In particular, SCION is already 
in production use since 2017, 
when a large Swiss bank switched 
a branch to only rely on SCION 
for communication with their data center. 
Since then, 7 ISPs are now commercially 
offering SCION connections~\cite{swisscom,anapaya}. 
An important aspect of today's SCION
production deployment is 
that it is operating independently of BGP, 
so it not an overlay network over BGP.

While PAN architectures have thus
experienced partial deployment, 
surprisingly little is known today 
about the \emph{interconnection agreements}
between ASes
possible in such architectures. These interconnection agreements, however, 
are highly relevant for both ASes and end-hosts.
From the perspective of ASes,
interconnection agreements determine
the economic opportunities offered 
by \pan{s},
which are critical to PAN adoption. From the perspective
of end-hosts, interconnection 
agreements play an essential role
for path diversity; the extent
of path diversity, in turn, 
influences the
magnitude of above mentioned
resilience and efficiency
improvements of \pan{s}.

In this context, we observe that \pan{} architectures 
enable new types of interconnection agreements that
are not possible in today's Internet.
Nowadays, interconnection agreements are heavily influenced
by the Gao--Rexford conditions 
(henceforth: GRC)~\cite{gao2001stable,cittadini2010assigning},
which prescribe that traffic from peers and 
providers must not be forwarded 
to other peers or providers.
It is important to distinguish 
between two aspects of the GRC 
which refer to independent concerns: 
a stability aspect and an economic
aspect. Regarding stability, the GRC provably 
imply route convergence of the Border Gateway Protocol (BGP)~\cite{gao2001stable},
and from an economic perspective,
the GRC signify that an AS only
forwards traffic if the cost of forwarding
can be directly recuperated from customer ASes or end hosts.

\begin{figure}
    \centering
    \begin{tikzpicture}[scale=2]
	
	\node[asnode] (A) at (.5,  1) {$A$};
	\node[asnode] (B) at (2.5,  1) {$B$};
	\node[asnode] (C) at (1.5, .5) {$C$};
	\node[asnode] (D) at (1,  0) {$D$};
	\node[asnode] (E) at (2,  0) {$E$};
	\node[asnode] (F) at (3,  0) {$F$};
	\node[asnode] (G) at (4,  0) {$G$};
	\node[asnode] (H) at (1, -.75) {$H$};
	\node[asnode] (I) at (2, -.75) {$I$};
	
	\foreach \prov/\cust in {A/C, A/D, B/E, B/F, B/G, D/H, E/I} \draw[->] (\prov) to (\cust);
	\foreach \p/\q in {A/B, C/D, C/E, D/E, E/F, F/G} \draw[dashed] (\p) to (\q);
	
	\draw[-,red,thick] (A.south east) .. controls (D) .. (E.north west);
	\draw[-,red,thick] (B.south west) .. controls (E) .. (D.north east);
	\draw[-,red,thick] (D.south east) -- (F.south west);
	\node[red] (a1) at (1.8, .5) {$a$};
	
	\draw[dashdotted,blue,thick] (A.south east) to[spline through={(D.east)(2,-0.2)}] (F.200);
	\draw[dashdotted,blue,thick] (E.north east) -- node[above,pos=.2] {$a'$} (G.north west);

\end{tikzpicture}
    \caption{AS topology with interconnection agreements $a$ and $a'$ (discussed in~\cref{sec:model:interconnection}). Peering links are shown as dashed lines, provider--customer links as ``provider $\rightarrow$ customer''.}
    \label{fig:model:new-examples}
\end{figure}
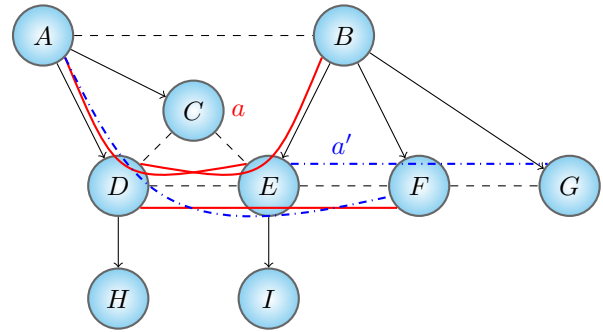

However, \pan{} architectures no longer require the GRC for providing stability. 
While paths in \pan{} architectures are discovered similarly as in
BGP, namely by communicating path information to neighboring ASes, 
data packets are forwarded along a path
selected by the packet source, which is embedded in the packet headers.
Thus, \pan{} architectures trivially solve convergence issues in the
sense of achieving a consistent view of the used forwarding
paths, as we will explain in~\cref{sec:convergence}. 

\pan{} architectures therefore
present the exciting opportunity to create
and use GRC-violating
paths---if such
paths can be made economically 
viable.
In particular, we observe that \pan{} architectures may
no longer require the GRC
for reasons of stability, but must still
respect the economic logic that
makes the GRC a rational forwarding 
policy. For example, while the creation
of GRC-violating path~$ADE$ by~$D$ in \cref{fig:model:new-examples} 
may not lead to convergence problems, the path
still is economically undesirable for~$D$,
because~$D$ would incur a charge from its 
provider~$A$ for forwarding traffic of~$E$,
which it cannot recuperate due to~$E$'s status as a 
peer. 

In this paper, we tackle this challenge
by proposing new interconnection agreements based on 
\emph{mutuality}, a concept that is 
already present in peering agreements today, 
but can be leveraged to set up more
complex and flexible agreements. Concretely, mutuality means that  the
mentioned example path~$ADE$ could be rendered economically viable
for $D$ by requiring a \textit{quid pro quo} from $E$,
the main beneficiary of the path.
For example, $E$ could offer path~$DEB$ to $D$
such that both~$D$ and~$E$ could save transit
cost for accessing ASes~$B$ and~$A$, respectively,
but incur additional transit cost for
forwarding their peer's traffic to their respective
provider. Moreover, ASes~$D$ and~$E$ might
as well provide each other with access to their
peers $C$ and~$F$, thereby saving
transit cost while experiencing additional load on
their network. 
If the flows over the new path segments
are properly balanced and especially if 
the new path segments allow ASes~$D$ and~$E$
to attract additional revenue-generating traffic from
their customers, such unconventional agreements
can be mutually beneficial. Hence, \pan{s}
offer opportunities for profit maximization
which are not present in today's Internet.

Concluding mutuality-based
agreements affects revenue and cost of the AS parties
in many ways, which requires a careful structuring of such agreements. We envisage that such agreements contain conditions
that must be respected in order to preserve the 
positive value of the agreement for both parties.
To this end, we further present a formal model of
AS business calculations and AS interconnections
that allows to derive two different types of agreement
conditions, namely conditions based on flow volumes
and conditions based on cash compensation.
Moreover, we show how to shape mutuality-based 
agreements to maximize the utility (i.e., the profit) 
obtained by both 
parties, and how
to negotiate them
efficiently.

Finally, we investigate the effect of mutuality-based
agreements on path diversity, 
building on a combination of several 
publicly available datasets~\cite{caida,caidaprefix,caidageo,geolite2}.
Our results underpin the benefits of mutuality-based agreements, 
which provide ASes access to thousands of additional paths,
many of which are considerably more attractive
regarding latency and bandwidth
than the previously available paths.
We believe that the
the highly beneficial
agreements investigated
in this paper
could represent a catalyst
to overcome the still
limited deployment of
PAN architectures.
\section{The Relevance of GRC for BGP and PANs}
\label{sec:convergence}

To clarify why the GRC are needed in a BGP/IP-based Internet but not
in \pan{} architectures,
we compare their convergence requirements using the example 
topology of~\cref{fig:model:new-examples}.

The fundamental issue with convergence in BGP
is the next-hop principle:
ASes can only select a next-hop AS for their traffic
and thus rely on that AS to forward the traffic
along the route that was originally communicated via BGP. 
If this assumption is violated%
---even temporarily---%
routing loops can arise. Put differently, in a BGP/IP-based 
Internet, all ASes need to share a common view of the
used forwarding paths.

Now suppose that ASes~$D$ and~$E$ forwarded routes from
their respective providers $A$ and~$B$ to each
other, which violates the GRC.
Assuming both $D$ and $E$ prefer routes learned from peers,
this results in a (slightly extended) instance of the classical DISAGREE
example~\cite{griffin1999analysis}, which does converge with BGP but non-deterministically. The non-determinism of such topologies, which are also known as ``BGP wedgies''~\cite{rfc4264}, is undesirable, but it does not
constitute a fundamental problem for convergence in BGP.
However, adding a single additional AS $C$, which
concludes similar agreements with both $D$ and $E$,
this topology leads to the famous BAD GADGET,
which has been shown to cause persistent route
oscillations~\cite{griffin1999analysis}.

This susceptibility to oscillations
is also worrisome because
seemingly benign topologies and policies may easily
reduce to the BAD GADGET in case one network link fails~\cite{griffin1999analysis}.
This shows that GRC-violating policies need to be implemented
very carefully and with coordination among all involved parties
to ensure routing stability
(e.g., using BGP
communities). 
As a consequence, ``sibling'' agreements
in which two ASes provide each other access to their respective providers
(as presented above) generally only exist between ASes controlled
by a single organization.

Unlike IP, \pan{s} forward a packet along the path encoded in its
header. Thus, there is no uncertainty about the traversed forwarding
path after the next-hop AS and routing loops
can be prevented. For example, if a source in $D$
would encode path~$DEBA$ in packets sent to a receiver in $A$,
$E$ would not send these packets back to $D$.
Precautions like the GRC are therefore not required for stability
in a \pan{} and ASes have substantially more freedom
in deciding which interconnection agreements to conclude and which paths
to authorize.

\section{Modelling Interconnection Economics}
\label{sec:model}

In this section, we describe our model of the economic
interactions between ASes in the Internet, which
allows to derive quantitative conditions that must be
fulfilled by interconnection agreements.

\subsection{AS Business Calculation}
\label{sec:model:business}

We model the Internet as a 
mixed graph~$\mathcal{G} = (\mathcal{A}, \mathcal{L}_{\leftrightarrow}, \mathcal{L}_{\uparrow} )$, 
where the nodes $\mathcal{A}$ correspond to ASes, the undirected 
edges $\mathcal{L}_{\leftrightarrow}$ correspond to peering links, and
the directed edges~$\mathcal{L}_{\uparrow}$
correspond to provider--customer
links. An edge $(X,Y)\in\mathcal{L}_{\uparrow}$ corresponds to a link from 
provider $X$ to customer $Y$. 
An AS node $X \in \mathcal{A}$ is connected to a set of 
neighbor ASes that can be decomposed into a provider 
set~$\pi(X)$, a peer set $\varepsilon(X)$, and a customer 
set~$\gamma(X)$.
For simplicity of notation,
we denote the customer
end-hosts of $X$ as
a virtual stub $\Gamma_X
\in \gamma(X)$, connected
over a virtual provider--customer
link~$\ell'$.

Each provider--customer link~$\ell = (X,Y) \in \mathcal{L}_{\uparrow}$ has a 
corresponding pricing function~$p_{\ell}(f_{\ell})$,
yielding the amount of money $X$ receives from $Y$
given flow volume~$f_{\ell}$ on link~$\ell$. This flow
volume~$f_{\ell}$ can
be interpreted as is
appropriate for the pricing
function, e.g., as the
median, average, or 95th
percentile of traffic volume
over a given time period.
Each pricing function~$p_{\ell}(f_{\ell})$ is of the
form~$\alpha_{\ell}f_{\ell}^{\beta_{\ell}}$,
where $\alpha_{\ell} \geq 0$ and~$\beta_{\ell} \geq 0$ are 
pricing-policy parameters. 
For example, $\beta_{\ell} = 0$ corresponds
to flat-rate pricing with flow-independent fee 
$\alpha_{\ell}$, $\beta_{\ell} = 1$ corresponds to
to pay-per-usage pricing with traffic-unit 
cost~$\alpha_{\ell}$, and $\beta_{\ell} > 1$ results in a 
superlinear pricing function, e.g., as given in congestion
pricing. For simplicity, we assume that all peering links~$\ell' \in \mathcal{L}_{\leftrightarrow}$
are settlement-free,
as usual in the 
literature~\cite{huston1999interconnection}. 
Paid-peering links
can be represented in the model as
provider--customer links. We write~$p_{XY} = p_{(X,Y)}$ for brevity.

In addition to charges defined by pricing functions,
an AS~$X$ incurs an internal cost according to an
internal-cost function $i_{X}(f_{X})$, which is non-negative and monotonously increasing in
the flow~$f_X$ through $X$. 
Furthermore, let the flow $f_{XY}$ be the share of $f_{X}$
that also flows directly to or from its neighbor $Y$.
These sub-flows are 
represented in 
vector~$\mathbf{f}_{X}$, i.e., 
$(\mathbf{f}_{X})_Y = f_{XY}$.
Analogously, $f_{XYZ}$ is the flow volume
on the path segment consisting of ASes~$X$, $Y$, and~$Z$ 
in that order, independent of direction. 

The utility (or profit) $\mathcal{U}_{X}(\mathbf{f}_{X}) =r_{X}(\mathbf{f}_X) - c_X(\mathbf{f}_X)$ of an AS $X$ is the difference between the revenue~$r_{X}(\mathbf{f}_X)$
obtained and the costs~$c_X(\mathbf{f}_X)$ incurred by traffic distribution~$\mathbf{f}_{X}$: 
\begin{subequations}%
    \begin{align}
        r_{X}(\mathbf{f}_X) &= \sum_{Y \in \gamma(X)} p_{XY}(f_{XY}),\\
        c_{X}(\mathbf{f}_X) &= i_X(f_X) + \sum_{Y \in \pi(X)} p_{YX}(f_{XY}).
    \end{align}
\end{subequations}

This simple model can already formalize some important
insights. For example, consider ASes~$A$, $D$, and $H$
in \cref{fig:model:new-examples}, 
connected by provider--customer
links~$(A,D)$ and $(D,H)$. For $D$ to make
a profit, i.e., $\mathcal{U}_{X}(\mathbf{f}_{X})>0$, it must hold that~$r_{D}(\mathbf{f}_{D}) > c_{D}(\mathbf{f}_D)$.
This in turn implies $p_{DH}(f_{DH}) + p_{D\Gamma_D}(f_{D\Gamma_D}) > p_{AD}(f_{AD}) + i_D(f_D)$, i.e., the revenue from $H$ and
$D$'s customer end-hosts must 
cover the cost induced by charges from $A$ as well as
internal cost.

\subsection{Interconnection Agreements}
\label{sec:model:interconnection}

We denote an interconnection agreement~$a$ 
between two ASes~$X$
and~$Y$ in terms of the respective neighbor
ASes to which
$X$ and $Y$ gain new paths thanks
to the agreement: \begin{equation}
    a = \big[X(\uparrow\pi'_X, \rightarrow\varepsilon'_X,\downarrow\gamma'_X); Y(\uparrow\pi'_Y, \rightarrow\varepsilon'_Y,\downarrow\gamma'_Y)\big]
\end{equation} Here $\pi'_X \subseteq \pi(X)$,
$\varepsilon'_X \subseteq \varepsilon(X)$ and
$\gamma'_X \subseteq \gamma(X)$ are the providers,
peers, and customers of AS $X$, respectively, 
to which $Y$ obtains access through the agreements
(analogous for $\pi'_Y$, $\varepsilon'_Y$, $\gamma'_Y$).
Furthermore, we introduce the notation
$a_X = \pi'_X \cup \varepsilon'_X \cup \gamma'_X$, 
and an analogous notation for $Y$.

Next, we formalize the utility of interconnection 
agreements.
Let the utility~$u_X(a)$ of agreement~$a$ to $X$
be the difference in $\mathcal{U}_{X}$ produced by
changes in flow composition due to agreement~$a$, i.e., 
\begin{equation}u_X(a) = \mathcal{U}_{X}(\mathbf{f}_{X}^{(a)}) - 
\mathcal{U}_{X}(\mathbf{f}_{X}) =  \Delta r_X - \Delta c_X,\end{equation} where $\mathbf{f}_{X}^{(a)}$ is the 
distribution of traffic passing 
through $X$ if agreement~$a$ 
is in force, and~$\Delta r_X$ 
and $\Delta c_X$ are 
agreement-induced changes in 
revenue and cost of $X$, 
respectively.

\subsubsection{Example of Peering Agreement}
\label{sec:model:interconnection:peering}

Consider the negotiation of a
classic peering agreement between 
ASes~$D$ and~$E$ in \cref{fig:model:new-examples},
which so far have been connected by 
providers $A$ and~$B$. 
We assume that ASes~$D$ and~$E$
are pure transit ASes, i.e.,
there are no customer end-hosts
within these ASes.
In a classic peering agreement,
both ASes provide each other access to all of their
respective customers. Using the notation introduced
above, this agreement is formalized as 
$a_p = [D(\downarrow\{H\});E(\downarrow\{I\})]$.
The change in revenue for $D$,
\begin{equation}\Delta r_D = p_{DH}(f_{DH}^{(a_p)}) - p_{DH}(f_{DH}),
\end{equation}
results from changes in flows to $D$'s customer $H$,
driven by the new peering link~$\ell'(D,E)$. The changes
in cost to $D$, 
\begin{equation}
\begin{split}
\Delta c_D &= 
i_D(f_D^{(a_p)}) - i_D(f_D) \\
&\hphantom{{}={}}+ p_{AD}(f_{AD}-f_{DABE}) - p_{AD}(f_{AD}),
\end{split}
\end{equation} 
result from changes in internal and provider 
cost. The utility of a peering agreement
to $D$ is then 
$u_D(a_p) = \Delta r_D - \Delta c_D$. The strongest 
rationale for peering agreements is that the agreement 
leads to considerable
cost decrease, i.e., a strongly negative~$\Delta c_D$,
as provider $A$ can be avoided for any traffic
$f_{DE}$. The new peering link may also attract
additional traffic from customer $H$ (e.g., due to the lower latency of the new connection),
thus increasing $D$'s revenue. 
If $\Delta r_D > \Delta c_D$,
agreement~$a_p$ has positive utility $u_D(a_p)>0$ and is
worth concluding from $D$'s perspective.
However, $D$ may also experience a substantial increase
in internal cost ($\Delta i_D$) due to peering, with
little savings in provider cost and no extra income 
from the additionally attracted traffic (e.g., due to flat-rate 
fees). In such a case, $u_D(a_p)$ is negative, and the
agreement is not attractive to $D$. 
For agreement~$a_p$ to be concluded, both~$u_D(a_p)$
and~$u_E(a_p)$ need to be non-negative (or, if
cash transfers as in
paid 
peering~\cite{zarchy2018nash} 
are used,
$u_D(a_p) + u_E(a_E)$ would
need to be non-negative
such that one party could
compensate the other party
and still benefit from
the agreement).

\subsubsection{Example of Novel Mutuality-Based Agreement}
\label{sec:model:interconnection:mutuality}

As discussed in \cref{sec:convergence}, the GRC are not necessary for stability in a \pan{}, 
which allows for new types of interconnection agreements.
In the example topology of \cref{fig:model:new-examples}, 
the following agreement~$a$ could
not be concluded in today's BGP-based Internet,
but could be concluded in a path-aware inter-domain
network: \begin{equation}a = [D(\uparrow\{A\});E(\uparrow\{B\},\rightarrow\{F\})]\end{equation} In this agreement~$a$, $D$ offers
$E$ access to its provider $A$, whereas $E$
in return provides $D$ with access to its provider 
$B$ and its peer $F$.

The agreement utility~$u_D(a)$ of agreement~$a$
to $D$ can be derived similar to the peering-agreement
example in \cref{sec:model:interconnection:peering}. Namely,
the changes in revenue and cost of $D$ are
\allowdisplaybreaks
\begin{subequations}%
    \begin{align}%
        \Delta r_D &= \sum_{X \in \gamma(D)} p_{DX}(f_{DX}^{(a)}) - p_{DX}(f_{DX}),\\
        \begin{split}
        \Delta c_D &=
                    i_D(f_D^{(a)}) - i_D(f_D)
                    + \smash{\smashoperator{\sum_{{Y \in \pi(D)}}}} p_{YD}(f_{DY}^{(a)}) - p_{YD}(f_{DY}),
        \end{split}
    \end{align} where
    \begin{equation}
        f_{DY}^{(a)} = f_{DY} + f_{EDY}^{(a)} - \sum_{Z\in a_{E}} f_{DY}^{\updownarrow}(Z,E)
    \end{equation}
\end{subequations} is the flow from $D$ to one of its
providers, $Y$, after conclusion of the agreement.
This flow towards the provider is increased 
by the traffic $f_{EDY}$ that $D$ transfers
to $Y$ for $E$ in accord with the agreement.
Simultaneously, flow $f_{DY}$ is decreased by the flow~
$f_{DY}^{\updownarrow}(Z,E)$ to the destinations $Z \in a_{E}$ 
that was previously forwarded via
provider $Y$, but is newly forwarded over $E$
thanks to agreement~$a$. As the new paths over agreement
partner $E$ are the reason for newly attracted traffic 
from $D$'s customers, all such newly attracted 
traffic is forwarded over the agreement partner, not over 
$D$'s providers.

Clearly, agreement~$a$ is not per se attractive
to ASes~$D$ and~$E$. For $D$, the higher the amount of 
flow from $E$ that newly must be forwarded to a provider
AS, the less attractive the agreement to $D$, i.e., 
the higher~$\Delta c$. In contrast, the higher the
amount of flow offloaded to $E$, the higher
the utility that $D$ can derive from agreement~$a$.
Vice versa, the agreement utility for $E$ 
is conversely affected by the size of these new flows.
Thus, the agreement~$a$ must be qualified. Necessarily,
these qualifications must guarantee positive agreement 
utility to both parties. Furthermore, it is desirable
that the qualifications achieve Pareto-optimal~\cite{debreu1954valuation} and fair
agreement utility, i.e.,
no party's utility could be increased without
decreasing the other party's utility, and the utility obtained by both parties
is as similar as possible. 
In \cref{sec:optimal},
we propose two different types of agreement
qualifications to achieve these goals.

\subsubsection{Extension of Agreement Paths}
\label{sec:model:extension}
Thanks to
agreement~$a$, ASes~$D$ and~$E$ obtain access to
the new path segments~$DEB$ and~$DEF$ ($D$) and $EDA$ ($E$). 
As the motivation behind the agreement
is the attraction of additional customer
traffic, the agreement parties would provide access
to the new path segments only to their respective 
customers. For example, $D$ would extend the new path 
segment~$DEB$ to~$HDEB$, but not to~$ADEB$ or~$CDEB$.

However, the new path segments can themselves 
become the matter of other agreements. For example,
in an agreement~$a'$ between $E$ and $F$,
$E$ could provide $F$ with access to path 
segment~$EDA$ if $F$ in return provides access
to its peer~$F$. Note that agreement~$a'$
must be negotiated such that the conditions defined
in agreement~$a$ can still be respected, as these
agreements are interdependent.

\section{Optimization
of Mutuality-Based
Agreements}
\label{sec:optimal}

The novel
mutuality-based agreements should achieve
Pa\-reto-op\-ti\-mal and fair utility 
in order to be attractive to both agreement
parties. Moreover, a necessary
economic condition for conclusion of the agreement is
the guarantee of non-negative agreement utility for 
both parties. Hence, defining an optimal interconnection
agreement between two ASes~$D$ and~$E$ 
corresponds to solving the nonlinear program 
\begin{equation}
\begin{array}{ll@{}ll}
\text{maximize}  & \displaystyle u_D(a) \cdot u_E(a) &\\
\text{subject to}& \displaystyle u_D(a) \geq 0, \quad &u_E(a) \geq 0,
\label{eq:optimal:general-program}
\end{array}
\end{equation} where the objective is given by
the Nash 
product~\cite{nash1950bargaining,binmore1986nash}, which
is only optimized for Pareto-optimal and fair
values of~$u_D(a)$ and~$u_E(a)$. Hence,
if the Nash bargaining product is optimized,
no party can increase its utility without decreasing
the other party's utility, and the utility of both
parties is as similar as possible.

In the following, we present two methods to solve
the nonlinear program 
in~\cref{eq:optimal:general-program}, i.e., two
ways to qualify agreement~$a$ such that the constrained
optimization problem is solved: The optimization
method in~\cref{sec:optimal:target-flow} is based
on defining flow-volume targets, which offers
better predictability,
whereas the method in 
\cref{sec:optimal:cash-compensation} is based
on cash transfers between the agreement parties,
which offers more 
flexibility.

\vspace*{2mm}
\subsection{Optimization via Flow-Volume Targets}
\label{sec:optimal:target-flow}

The constrained optimization problem 
in~\cref{eq:optimal:general-program} can be solved
by determining volume limits for the flows
that traverse the new path segment created by agreement~$a$. 
Concretely,
the general optimization problem can be instantiated
by the nonlinear program
\begin{align}
&\text{max} \quad \displaystyle u_D(\mathbf{f}^{(a)}, \Delta \mathbf{f}^{(a)}) \cdot u_E(\mathbf{f}^{(a)}, \Delta \mathbf{f}^{(a)})\notag\\
&\text{s.t.} \quad \displaystyle \Delta r_D(\mathbf{f}^{(a)}, \Delta \mathbf{f}^{(a)}) \geq \Delta c_D(\mathbf{f}^{(a)}, \Delta \mathbf{f}^{(a)}) & \text{(I-D)}&\notag\\ 
&\quad \forall X \in a_{E}.\ f_{DEX}^{(a)} \geq \textstyle\sum_{Z \in \gamma(D)} \Delta f_{ZDEX}^{(a)} & \text{(II-D)}&\notag\\
&\quad \forall X \in a_E.\ \forall Z \in \gamma(D). \  \Delta f_{ZDEX}^{(a)} \leq \Delta f_{ZDEX}^{\max} & \text{(III-D)}&\notag\\
&\quad + \text{ constraints (I-E), (II-E), (III-E) where \rlap{$D\leftrightarrow E$.}}
\label{eq:optimal:target-program}%
\end{align}
Here, 
$f^{(a)}_P$ refers to the
total flow volume on a new path segment~$P$ allowed by the
agreement, and~$\Delta f_P^{(a)}$ is the volume of newly attracted customer traffic
on path segment~$P$ after agreement conclusion.
Hence, the flow volume
on new path segment~$P$ that consists of rerouted
existing traffic is at most $f^{(a)\updownarrow}_P = f^{(a)}_P - \Delta f^{(a)}_P$.

The constraints (I-D) and (I-E) capture the fact
that the agreement must be economically viable
for both parties. The constraints (II-D) and (II-E)
capture the requirement that all the agreement-induced additional traffic from customers
has to be accommodated within the flow
allowances defined in the agreement.
Finally, as any agreement could be made viable
by attracting enough additional customer
traffic, the constraints (III-D) and (III-E)
express that there is a limit $\Delta f_{P}^{\max}$
to customer demand for new path segment~$P$.
The optimization problem can be solved
by appropriately adjusting~$f^{(a)}$, i.e.,
the total allowance for flows on the new path 
segments, and~$\Delta f^{(a)}$, i.e., the amount
of additionally attracted customer traffic on the new
path segments. The resulting~$f^{(a)}$
can then be included into the agreement as
flow-volume targets, whereas the 
resulting~$\Delta f^{(a)}$ can be used by each AS
to optimally allocate the flow-volume
allowance among its customers. If agreement
paths are extended as
discussed in~\cref{sec:model:extension},
additional constraints
may hold; in this paper,
we do not investigate
these constraints.

\subsection{Optimization via Cash Compensation}
\label{sec:optimal:cash-compensation}

An alternative optimization method to fixing
flow-volume targets is given by a non-technical
approach that is based on cash transfers
between the agreement parties. The idea
of such an agreement structure is to abstain
from limiting flow volumes, but to agree upon
a cash payment~$\pi$ for compensating the party
that benefits less or even stands to lose
from the agreement. Formally, negotiating
an agreement between ASes~$D$ and~$E$ is
equivalent to defining a cash sum
$\Pi_{D\rightarrow E}$ from $D$ to $E$ (for 
negative~$\Pi_{D\rightarrow E}$,
$E$ pays $D$)
that solves the optimization problem
\begin{equation}
\begin{array}{@{}ll@{}ll}
\text{max}  & \displaystyle \big(u_D(a) - \Pi_{D\rightarrow E}\big) \big(u_E(a) + \Pi_{D\rightarrow E}\big) &\\
\text{s.t.}& \displaystyle u_D(a) - \Pi_{D\rightarrow E} \geq 0, \quad u_E(a) + \Pi_{D\rightarrow E} \geq 0.
\label{eq:optimal:cash-program}
\end{array}%
\end{equation}
In negotiation, the
utilities $u_D(a)$ and~$u_E(a)$ are estimated
based on the expected volume of the newly enabled
flows.

The optimization problem in~\cref{eq:optimal:cash-program}
has a solution if and only if $u_D(a) + u_E(a) \geq 0$,
i.e., one party gains at least as much as the other
party loses and can thus compensate the losing party
while still benefiting from the agreement.
If $u_D(a) + u_E(a) \geq 0$, \cref{eq:optimal:cash-program}
is always solved based on the
Nash Bargaining Solution~\cite{nash1950bargaining}:
\begin{equation}\Pi_{D\rightarrow E} = 
u_{D}(a) - \frac{u_D(a) + u_E(a)}{2}.\end{equation}

\subsection{Comparison of Optimization Methods}
\label{sec:optimal:comparison}

The main advantage of flow-volume targets
over cash transfers is their higher predictability:
As flow-volume agreements allow the agreement parties
to enforce volume limits, they are
more likely to guarantee positive agreement utility
than cash-compensation agreements. The latter depend
on ex-ante estimates of newly attracted
customer traffic that might
be incorrect, in which case the
stipulated cash sum might not respect the 
constraints in~\cref{eq:optimal:cash-program}. 

Besides being easier to compute, an important advantage of 
cash compensation over volume targets is
its larger flexibility, which
translates into higher probability of the agreement
being concluded as well as higher achievable 
joint utility. 
In certain settings where the revenues and costs of two ASes
are very dissimilar, the flow-volume optimization
problem in~\cref{eq:optimal:target-program}
has a solution where all flow-volume targets are
zero, i.e., the agreement cannot be concluded.
In contrast, a cash-compensation agreement can always be concluded
as long as the joint utility is positive.

A common difficulty of both agreement structures
is that they depend on private information
of the negotiating parties, namely
the charges from their respective providers,
their internal forwarding cost, and the pricing
for their customers, which determine
the utility each party derives from the agreement.
It cannot be assumed that the parties
are willing to truthfully reveal this private information,
as false claims about the
cost structure strengthen a party's bargaining
position. In~\cref{sec:negotiation},
we show how the inefficiency arising from
such private information can be limited
by means of a bargaining mechanism.
\section{Mechanism-Assisted Negotiation}
\label{sec:negotiation}

In this section, we present a bargaining
mechanism that we have
designed to allow two interested parties
to negotiate a mutuality-based interconnection
agreement in an automated fashion, 
while reducing the negotiation
inefficiency arising from bargaining
under private information. However,
while
there are considerable advantages to using
such a bargaining mechanism, there is no
inherent necessity to use it; 
mutuality-based agreements might as well
be negotiated by classic offline negotiations
similar to classic peering agreements.

\subsection{Problem Statement}
\label{sec:negotiation:problem}

When negotiating a mutuality-based
interconnection agreement~$a$, the agreement
parties~$X$ and~$Y$ must agree on flow-volume
targets or a cash sum transferred
between the parties. The core difficulty
of such negotiations is that the
determination of the agreement conditions
relies on~$u_X(a)$ and~$u_Y(a)$, i.e.,
the amount of utility that
either party derives from the agreement,
which is unknown to the respective
other party.
The presence of such private
information allows each party to
falsely report a lower agreement utility
than it really obtains, which leads
to more favorable terms of the agreement
for the dishonest party. For example,
when negotiating a cash-compensation
agreement, the after-negotiation utility~$\overline{u}_X$ 
of party~$X$ is determined by
\begin{equation}
\overline{u}_X = u_X - \Pi_{X\rightarrow Y} = u_X - \frac{v_X - v_Y}{2},
\end{equation} where~$v_X$ and~$v_Y$
are the values of the utility which~$X$
and~$Y$ \emph{claim} to obtain from the agreement
and which are used for determining the 
cash-compensation sum~$\Pi_{X\rightarrow Y}$.
To simplify our notation, we drop the reference
to $a$ here and in the remainder of the section,
as we always consider a single agreement.
Clearly, party~$X$ can 
increase~$\overline{u}_X$ by 
decreasing~$v_X$, i.e., its utility claim.
However, if
both parties follow such a dishonest
strategy, the apparent utility 
surplus~$v_X + v_Y$
of the agreement tends to become
negative, in which case the agreement
seems to be not worth concluding,
the negotiation breaks down and
both parties derive zero utility. 
Hence, the challenge in negotiating 
mutuality-based agreements 
(as for paid-peering agreements
in today's 
Internet~\cite{zarchy2018nash}) is 
the classic problem 
of \emph{non-cooperative
bilateral bargaining}~\cite{myerson1983efficient}.

In the game-theoretic literature,
this problem is tackled by \emph{mechanism
design}, i.e., by structuring the negotiation
in a way that minimizes the inefficiency
of the result. Especially
for inter-AS negotiation,
such mechanisms have the additional
advantage that they enable the automation
and mathematical characterization of
negotiations which nowadays are often informal
and risky~\cite{calvert2019vision}. In this
section, we present a bargaining mechanism 
with multiple desirable properties.
We focus on negotiating cash-compensation
agreements 
(cf.~\cref{sec:optimal:cash-compensation}); 
adapting the mechanism
for flow-volume
agreements (cf.~\cref{sec:optimal:target-flow})
is an interesting challenge for
future work.

\subsection{Desirable Mechanism Properties}
\label{sec:negotiation:properties}

Typically, desirable properties of
bilateral-bargaining mechanisms include
the following~\cite{myerson1983efficient}:
\begin{enumerate}[props]
    \item \textit{Individual rationality:}
    \label{prop:individual_rationality}
    Participation in the mechanism should be
    associated with non-negative
    utility in expectation (weak individual rationality)
    or in any outcome (strong individual rationality) for any party
    such that no party must be forced
    to take part in the mechanism.
    \item \textit{Ex-post efficiency:}
    \label{prop:ex-post_efficiency}
    The mechanism should lead to conclusion
    of the agreement if and only if the
    utility surplus is non-negative.
    \item \textit{Incentive compatibility:}
    \label{prop:incentive_compatibility}
    The mechanism should structure the negotiation
    such that it is in a party's self-interest
    to be honest about its valuation
    of the agreement.
    \item \textit{Budget balance:}
    \label{prop:budget_balance}
    The mechanism should neither require
    external subsidies nor end up with
    left-over resources (e.g., money) that 
    are not ultimately assigned to 
    the negotiating parties \cite{nath2019efficiency}.
\end{enumerate}

According to the famous
Myerson--Satterthwaite
theorem,
no mechanism can satisfy
the requirements \labelcref{prop:budget_balance,prop:individual_rationality,prop:ex-post_efficiency}
simultaneously~\cite{nachbar2017myerson,myerson1983efficient}. The
prominent Vickrey--Clarkes--Grove 
(VCG) mechanism, for example,
guarantees individual rationality (\labelcref{prop:individual_rationality}) and
ex-post efficiency (\labelcref{prop:ex-post_efficiency}), but violates
budget balance (\labelcref{prop:budget_balance}) \cite{vickrey1961counterspeculation, clarke1971multipart, groves1973incentives}.
Absent government
intervention, individual rationality
and budget balance are necessary
conditions for an inter-AS negotiation
mechanism; we therefore sacrifice perfect
ex-post efficiency and instead aim at maximizing
the Nash bargaining 
product: \begin{equation}
    \mathcal{N}(u_X, u_Y, v_X, v_Y) = \big(u_X - \Pi_{X\rightarrow Y}\big)\big(u_Y + \Pi_{X\rightarrow Y}\big)
\end{equation} if $v_X + v_Y \geq 0$, 
and 0 otherwise, where the 
cash transfer is~$\Pi_{X\rightarrow Y}=
(v_X - v_Y)/2$.

While there are bargaining mechanisms that offer
individual rationality, budget balance
and incentive compatibility (according
to the notion of Bayes--Nash incentive
compatibility (BNIC)~\cite{muller2007weak}),
there are three arguments for relaxing
the incentive-compatibility requirement
as well. First, incentive compatibility
might not be desired, because an AS might
not want to disclose its true utility
from an agreement for privacy reasons
(e.g., not to hamper its prospects
in future negotiations) and an
incentive-compatible mechanism would allow
the other party to learn the utility of
a party~$X$ from the mechanism-induced
truthful claim~$v_X = u_X$.
Second, incentive compatibility is in
general unnecessary to achieve an optimal
Nash bargaining product: For a viable agreement
(i.e.,~$u_X + u_Y \geq 0$), the Nash bargaining
product is optimized for all~$v_X$, $v_Y$
with~$v_X - u_X = v_Y - u_Y$ (i.e.,
equal dishonesty) and $v_X + v_Y \geq 0$.
Hence, while truthfulness, i.e.,
$v_X - u_X = 0 = v_Y - u_Y$, is a
sufficient condition for an optimal
Nash bargaining product, it is not
a necessary condition. Third,
while incentive compatibility can
be guaranteed with mechanisms,
this guarantee often comes at the
cost of introducing inefficiency:
For example, the randomized-arbitration
mechanism by Myerson~\cite{myerson1979incentive}
introduces a relatively high probability
of negotiation cancellation such that
even agreements with large surplus
often cannot be concluded. Counter-intuitively,
mechanisms which allow small deviations
from truthfulness might thus be more
efficient than perfectly incentive-compatible
mechanisms. In the following subsection,
we present such a mechanism.

\subsection{Bargaining in One Shot with Choice Optimization (BOSCO)}
\label{sec:negotiation:mechanism}

In this subsection, we present the BOSCO
mechanism, which we have
designed to enable automated
negotiation of inter-AS
agreements with high
bargaining efficiency.
The core idea of the BOSCO mechanism
is as follows: The negotiating parties
play a simple bargaining game
supervised by an BOSCO service, in which
each of them has a set of choices defined
by the mechanism. Each combination
of such choice sets is associated with
at least one Nash equilibrium, i.e.,
a combination of strategies in which
no party can profitably deviate from
the strategy assigned to it.
In turn, each such Nash equilibrium
can be rated with respect to a
bargaining-efficiency metric.
The benefit of the mechanism is
thus realized by the BOSCO service,
which appropriately constructs the choice
sets and picks an associated
Nash equilibrium such that a
high bargaining-efficiency results.
In the following, we will present
and formalize the components
of the mechanism.

\subsubsection{Utility Distributions}
For executing the BOSCO mechanism,
the two agreement parties~$X$ and~$Y$ 
communicate the content of a mutuality-based
agreement to an BOSCO service.
While the BOSCO service does not
know the true utility~$u_X$ and~$u_Y$
that either party derives from the
agreement, we assume (as usual
in bargaining-mechanism design) 
that the BOSCO service can estimate a
\emph{utility distribution}~$\mathbb{U}_Z(u)$,
which yields the probability that
party~$Z \in \{X,Y\}$ derives utility~$u$ from
the agreement. We envision that such
an estimation can be performed on 
the basis of heuristics, taking
standard transit and network-equipment
prices into account.

\subsubsection{Choice Sets} 
After deriving~$\mathbb{U}_Z(u)$
for each agreement party~$Z$, the BOSCO
mechanism constructs a \emph{choice set}~$V_Z$ of
possible claims for each agreement
party~$Z$. For BOSCO,
these choice sets correspond to 
finite discrete sets
with cardinality~$W_Z = |V_Z|$.
To guarantee strong rationality,
each choice set always contains
the option~$-\infty$, with which
any party can cancel the negotiation.
Moreover, let there be
an ordering~$v_{Z,1}, \dots, v_{Z,W_Z}$
on the choices such 
that~$v_{Z,i} < v_{Z,j}$
for all~$1\leq i < j \leq W_Z$.

\subsubsection{Bargaining Game}
The utility distributions
and the choice sets represent the basis of
a \emph{bargaining game}. In this bargaining
game, each party~$Z$ picks a suitable 
choice~$v_Z \in V_Z$ and commits 
it to the BOSCO service.
The BOSCO service then checks 
whether the apparent utility 
surplus is non-negative, i.e., $v_X + v_Y \geq 0$.
If yes, the mechanism 
prescribes conclusion of the agreement
with cash compensation~$\Pi_{X\rightarrow Y} = (v_X - v_Y)/2$, resulting in after-negotiation
utility $\overline{u}_X = u_X - \Pi_{X\rightarrow Y}$ and $\overline{u}_Y = u_Y + \Pi_{X\rightarrow Y}$. If not, the mechanism
cancels the negotiation, resulting
in after-negotiation 
utility~$\overline{u}_X = \overline{u}_Y = 0$.

\subsubsection{Bargaining Strategies}
In this bargaining game, the
\emph{bargaining strategy}~$\sigma_Z(u_Z)$ 
of party~$Z$ is a function which yields
a choice~$v_Z \in V_Z$ given 
the true utility~$u_Z$ of party~$Z$.
Party~$X$'s \emph{best-response strategy}~$\sigma_X^{+}$ to
party~$Y$'s strategy~$\sigma_Y(u_Y)$ 
consists of picking the 
choice~$v_X \in V_X$ with
the highest expected after-negotiation utility 
given its true utility~$u_X$, i.e.,
the choice~$v_X \in V_X$ that maximizes
\begin{equation} 
    \mathbb{E}[\overline{u}_X](u_X, v_X) = \smash{\smashoperator[l]{\sum_{\substack{v_Y \in V_Y.\\v_Y \geq -v_X}}}}\mathbb{P}\big[v_Y\big] \cdot
    \Big(u_X - \frac{v_X - v_Y}{2}\Big),
    \label{eq:negotiation:best-choice}
\end{equation} where
\begin{equation}
    \mathbb{P}\big[v_Y\big] = \int_{-\infty}^{\infty} \mathbb{U}_Y(u) \cdot [\![\sigma_Y(u) = v_Y]\!] \d{u}
\end{equation} and~$[\![P]\!]$ is 1
if statement~$P$ is true and 0 otherwise. 
Interestingly, the expected after-negotiation
utility given a choice~$v_X \in V_X$ is
a linear function~$m_X(v_X)u_X + q_X(v_X)$
of the true utility~$u_X$, with
\begin{align}
    m_X(v_X) &= \mathbb{P}[\sigma_Y(u_Y) \geq -v_X],\label{eq:negotiation:m}\\
    q_X(v_X) &= \smashoperator[l]{\sum_{v_Y \geq -v_X}} \mathbb{P}[\sigma_Y(u_Y) = v_Y] (v_Y - v_X)/2 ,
\end{align} where it holds
that~$m_X(v_X) \geq m_X(v_X')$
for all $v_X' < v_X$ by nature
of the CCDF in~\cref{eq:negotiation:m}.
For brevity, we 
write~$m_{X,i} = m_X(v_{X,i})$
(analogous for~$q_X$).
This linear formulation
allows an easy computation
of the best-response strategy:
A choice~$v_{X,i}$ is the best choice given
true utility~$u_X$ if it holds that
$m_{X,i}u_X + q_{X,i} \geq
m_{X,j}u_X + q_{X,j}$ 
for all~$j \neq i$.
If there exists~$j$ with~$m_{X,j} =
m_{X,i}$ and~$q_{X,j} > q_{X,i}$,
the choice~$v_{X,i}$ is not optimal for
any~$u_X$.
Otherwise, $v_{X,i}$ is the best choice
for~$u_X$ only if
\begin{equation}
    u_X \geq \smash{\frac{q_{X,j^-} - q_{X,i}}{m_{X,i} - m_{X,j^-}}} =: I(i, j^-)
\end{equation}
for all~$j^- \in J^-(i) =
\{j\ |\ j<i \land m_{X,j}\neq m_{X,i}\}$.
Let~$u_X^-(i) = \max_{j^{-} \in J^{-}(i)} I(i, j^{-})$ 
if~$|J^{-}(i)| > 0$,
and~$-\infty$ otherwise. 
Symmetrically,~$v_{X,i}$ is the best choice
for~$u_X$ only if~$u_X \leq I(i,j^+)$ 
for all~$j^+ \in J^{+}(i)$,
which is a set defined analogously
to~$J^{-}(i)$. Also analogously,
let~$u_X^+(i) = \min_{j^{+} \in J^{+}(i)} I(i, j^{+})$ if
$J^{+}(i)$ is non-empty, and
$\infty$ otherwise. Then,
all~$u_X$ for which~$v_{X,i}$ is the
best choice lie in the 
range~$[u_X^{-}(i), u_X^{+}(i)]$. 
Hence, the best-response 
strategy~$\sigma_X^+$ is defined
by a series of
thresholds~$\{t_{X,i}\}_{i \in \{1, ..., W_X+1\}}$,
where~$\sigma_X^+(u_X) = v_{X,i}$
if~$u_{X} \in [t_{X,i},t_{X,i+1})$.
\Cref{alg:best-response} shows
how to compute this threshold
series. The loop
at \cref{alg:thresholds.loop} is needed to generate
intervals which are disjoint,
but cover all possible values.

\begin{algorithm}[t!]
    \caption{Best-response computation for
    party~$X$.}
    \label{alg:best-response}
    \small
    \begin{algorithmic}[1]
        \Procedure{ComputeBestResponse}{$\{(m_{X,i}, q_{X,i})\}_i$}
        \State $t_{X,1} \gets -\infty$
        \State $t_{X,i} \gets \infty\  \forall i \in \{2, ..., W_X +1\}$
        \State $i \gets 1$
        \While{$|J^+(i)| > 0$}
            \State $i^+ \gets \arg\min_{j^+ \in J^{+}(i)} I(i,j^+)$
            \State $t_{X,i^+} \gets I(i,i^+)$
            \State $i \gets i^+$
        \EndWhile
        \For{$i = 1$ \textbf{to} $W_X$}\label{alg:thresholds.loop}
            \If{$t_{X,i} = \infty$}
                \State $t_{X,i} \gets \min_{j > i} t_{X,j}$
            \EndIf
        \EndFor
        \Return{$t_X$}
        \EndProcedure
    \end{algorithmic}
\end{algorithm}

\subsubsection{Nash Equilibria}
A \emph{Nash equilibrium}~$\sigma^{\ast} = (\sigma^{\ast}_X,
\sigma^{\ast}_Y)$ in the bargaining game
is a set of two bargaining strategies,
each of which is the best-response
strategy to the other strategy.
An equilibrium can be computed
by assuming arbitrary~$\sigma_X$ 
and~$\sigma_Y$
(defined by arbitrary threshold
series~$t_X$ and~$t_Y$) and
then compute
the best-response strategies in
an alternating fashion until
the best-response strategy of
any party is their existing
strategy. While it can
be shown that the considered
bargaining game is not a
potential 
game~\cite{monderer1996potential} (which would guarantee
convergence to an equilibrium
by alternating unilateral
optimization), the
best-response dynamics always
converged to an equilibrium
in our diverse simulations.

\subsubsection{Bargaining Efficiency}
Given a Nash equilibrium, a natural question
arises concerning the efficiency of such
an equilibrium~$\sigma^{\ast}$. Clearly, if the BOSCO
service knew~$u_X$ and~$u_Y$, it could
simply compute the associated Nash
bargaining 
product~$\mathcal{N}(u_X, u_Y, \sigma^{\ast}_X(u_X), \sigma^{\ast}_Y(u_Y))$ and compare
it to the optimal Nash bargaining 
product
$\mathcal{N}(u_X, u_Y, u_X, u_Y)$ that
arises under universal truthfulness.
However, as the BOSCO service has only
probabilistic knowledge about the true
utility of the agreement,
it must evaluate the
efficiency of an equilibrium~$\sigma^{\ast}$ 
by computing the the \emph{expected
Nash bargaining 
product}~$\mathbb{E}\left[\mathcal{N}\middle|\sigma^{\ast}\right]$ for this strategy, which is
\begin{equation}
    \smash{\iint_{-\infty}^{\infty}} \mathbb{U}(u_X,u_Y)\mathcal{N}\big(u_X,u_Y,\sigma^{\ast}_X(u_X), \sigma^{\ast}_Y(u_Y)\big) \d{u_Y}\d{u_X}
\end{equation} where~$\mathbb{U}$ is the
joint utility distribution for parties~$X$ and~$Y$.
The optimal expected Nash bargaining
product is given by~$\mathbb{E}\left[\mathcal{N}\middle|\sigma^{\top}\right]$ 
where $\sigma_Z^{\top}(u_Z) = u_Z$ is
the truthful strategy for party~$Z$.
Similar to a Price of Anarchy formulation~\cite{roughgarden2015intrinsic},
we thus formalize the efficiency of
an equilibrium with the following
metric of \emph{Price of Dishonesty} ($\mathit{PoD}$):
\begin{equation}
    \mathit{PoD}(\sigma^{\ast}) = 1 - \frac{\mathbb{E}\left[\mathcal{N}\middle|\sigma^{\ast}\right]}{\mathbb{E}\left[\mathcal{N}\middle|\sigma^{\top}\right]}
\end{equation} which is always in $[0,1]$
for reasons laid out in~\cref{sec:negotiation:provable-properties}. Note that~$\mathit{PoD}$
is undefined 
if~$\mathbb{E}[\mathcal{N}(\sigma^\top)] = 0$, i.e., if the agreement would
be consistently unviable even
under honesty, which is an 
uninteresting case that we 
henceforth disregard.

In summary, the BOSCO service is thus tasked with
estimating~$\mathbb{U}_X$ and $\mathbb{U}_Y$ and
constructing choice sets~$V_X$ and~$V_Y$ such that
the thereby defined bargaining game has an equilibrium~$\sigma^{\ast}$ 
with a low~\textit{PoD}. After the BOSCO service found such a
configuration, it communicates the 
mechanism-information 
set~$(\mathbb{U}_X,\mathbb{U}_Y,V_X,V_Y,\sigma^{\ast})$, 
to the communicating parties, which can verify
that~$\sigma^{\ast}$ is indeed a Nash equilibrium
and thus indeed follow the strategy that
is assigned to them in the equilibrium.
Hence, each party~$Z$
plays the bargaining game
by applying the equilibrium
strategy~$\sigma^{\ast}_Z$
to its true utility value~$u_Z$
and commits the resulting
claim to the BOSCO service,
which then decides on
agreement conclusion and,
in case of negotiation success,
on the exchanged cash compensation~$\Pi_{X\rightarrow Y}$.

\subsection{BOSCO Properties}
\label{sec:negotiation:provable-properties}

After describing the BOSCO mechanism in
the previous section, we now analyze the
mechanism with respect to the properties
listed in~\cref{sec:negotiation:properties}.
First of all, it is clear that the BOSCO
mechanism is budget-balanced, as the
cash transfer paid by one party is
exactly the cash transfer received by
the other party. We prove other
mechanism properties below.

\begin{theorem}
    The BOSCO mechanism offers \textbf{strong individual rationality},
    i.e., \begin{equation}\forall u_X,u_Y.\ \overline{u}_X \geq 0 \text{ and } \overline{u}_Y \geq 0.\end{equation}
    \label{thm:strong-rationality}
\end{theorem}
\begin{proof}
    Given utility~$u_X$ of party~$X$, $\sigma^{*}_X(u_X)$  is the best choice for party~$X$.
    If~$u_Y$ is such that~$\sigma^{\ast}_Y(u_Y) < -\sigma^{*}_X(u_X)$, 
    then the agreement is not concluded and~$\overline{u}_X = \overline{u}_Y = 0$.
    Conversely, if~$\sigma^{\ast}_Y(u_Y) \geq -\sigma^{\ast}_X(u_X)$,
    the agreement is concluded and~$\sigma^{\ast}_Y(u_Y)$ appears 
    in~$\mathbb{E}[\overline{u}_X](u_X, \sigma^{\ast}_X(u_X))$ from~\cref{eq:negotiation:best-choice}.
    Now assume that~$\overline{u}_X = u_X - (v_X - \sigma^{\ast}_Y(u_Y))/2 < 0$.
    In that case, $u_X - (v_X - v_Y)/2$ is negative 
    for any~$v_Y \in V_Y$ with $-v_X \leq v_Y \leq \sigma^{\ast}_Y(u_Y)$.
    Hence, $\mathbb{E}[\overline{u}_X]$ could be
    increased by choosing~$v_X' < v_X \in V_X$ with~$-v_X' > \sigma^{\ast}_Y(u_Y) \geq -v_X$
    such that all the mentioned~$v_Y$ drop from the objective,
    as the summands associated with these~$v_Y$ only contribute negative values
    and the summands associated with the remaining choices of party~$Y$ increase.
    Thanks to $-\infty \in V_X$, such a choice is always possible.
    However, this non-optimality of~$\sigma^{\ast}_X(u_X)$ contradicts
    the best-response character of~$\sigma^{\ast}_X$,
    which invalidates the assumption of a negative~$\overline{u}_X$.
    Hence, if an agreement is concluded, $\overline{u}_X \geq 0$ 
    for any~$u_X$ and~$u_Y$ (The case for~$\overline{u}_Y$ is
    analogous). In summary, such non-negativity of after-negotiation utility
    exists in any case (non-conclusion and conclusion),
    establishing strong individual rationality.
\end{proof}

\begin{theorem}
    The BOSCO mechanism is \textbf{sound}, i.e., it never leads to conclusion of a
    non-viable agreement:\footnote{In order to be ex-post efficient, 
    the mechanism would additionally need to be \emph{complete} in the sense that
    all viable agreements are concluded. However, as noted in~\cref{sec:negotiation:properties},
    this property is not achievable given other desirable properties.} 
    \begin{equation}\forall u_X,u_Y.\ \sigma^{\ast}_X(u_X) + \sigma^{\ast}_Y(u_Y) \geq 0 \implies u_X + u_Y \geq 0 \end{equation}
    \label{thm:soundness}
\end{theorem}
\begin{proof}
    If~$\sigma^{\ast}_X(u_X) + \sigma^{\ast}_Y(u_Y) \geq 0$, 
    the agreement is concluded and~$\overline{u}_X = u_X - \Pi_{X\rightarrow Y}$
    and~$\overline{u}_Y = u_Y + \Pi_{X\rightarrow Y}$. By strong
    rationality (\cref{thm:strong-rationality}), it holds 
    that~$\overline{u}_X \geq 0$ and~$\overline{u}_Y \geq 0$,
    which implies~$u_X \geq \Pi_{X\rightarrow Y}$ and
    $u_Y \geq -\Pi_{X\rightarrow Y}$. By addition
    of the inequalities, $u_X + u_Y \geq \Pi_{X\rightarrow Y} - \Pi_{X\rightarrow Y} =  0$.
    Hence, $\sigma^{\ast}_X(u_X) + \sigma^{\ast}_Y(u_Y) \geq 0 \implies u_X + u_Y \geq 0$.
\end{proof}

\begin{theorem}
    The BOSCO mechanism always leads to an equilibrium~$\sigma^{\ast}$
    with~$\mathit{PoD}(\sigma^{\ast}) \in [0, 1]$.
\end{theorem}
\begin{proof}
    We show that~$\mathcal{N}^{\ast} := \mathcal{N}(u_X, u_Y, \sigma^{\ast}_X(u_X), \sigma^{\ast}_Y(u_Y)) \leq
    \mathcal{N}(u_X, u_Y, u_X, u_Y) =: \mathcal{N}^{\top}$ for all~$u_X$, $u_Y$,
    which implies the theorem. 
    If~$u_X + u_Y < 0$, the agreement is neither concluded under truthfulness nor,
    by soundness (\cref{thm:soundness}), given the equilibrium strategy, 
    resulting in~$\mathcal{N}^{\ast} = \mathcal{N}^{\top} = 0$.
    Conversely, if~$u_X + u_Y \geq 0$, 
    $\mathcal{N}^{\ast} = 0 \leq \mathcal{N}^{\top}$
    if~$\sigma^{\ast}_X(u_X) + \sigma^{\ast}_Y(u_Y) < 0$ and
    $\mathcal{N}^{\ast} \leq \mathcal{N}^{\top}$ 
    if~$\sigma^{\ast}_X(u_X) + \sigma^{\ast}_Y(u_Y) \geq 0$
    by optimality of~$\mathcal{N}^{\top}$. Note that by strong
    rationality (\cref{thm:strong-rationality}), the
    case where~$\sigma^{\ast}_X(u_X) + \sigma^{\ast}_Y(u_Y) \geq 0$,
    $\overline{u}_X,\overline{u}_Y < 0$ and $\overline{u}_X\overline{u}_Y > \mathcal{N}^{\top}$
    is not possible.
\end{proof}

\begin{theorem}
    The BOSCO mechanism is \textbf{privacy-preserving}, i.e., an exact
    reconstruction of the true utility of a party from its choice 
    is impossible: \begin{equation}\forall v_X \in V_X.\ |\{u_X|\sigma^{\ast}_X(u_X) = v_X\}| \neq 1\end{equation}
\end{theorem}
\begin{proof}
    In terms of privacy, the only
    problematic case would arise
    if the range~$[t_{X,i}, t_{X,i+1})$
    contained only one value, which
    would allow the derivation
    of that value~$u_X$ from the
    associated choice~$v_{X,i}$.
    However, since half-open intervals
    on the real numbers cannot
    contain a single value, this
    problematic case cannot arise.
\end{proof}

While exact reconstruction
of the true utility is thus
impossible, it might
still be possible to
predict the true utility
with high precision if
the interval associated
with a choice is very
small. Hence, the degree
to which an equilibrium
preserves privacy could
be quantified (e.g., by
the length of the shortest
non-empty interval) and
then taken into account
by the BOSCO service
when picking an
equilibrium.

\subsection{Choice-Set Construction}
\label{sec:negotiation:optimization}

\begin{figure}[tb]
    \centering
    \includegraphics[width=\columnwidth, trim=0 0 0 0]{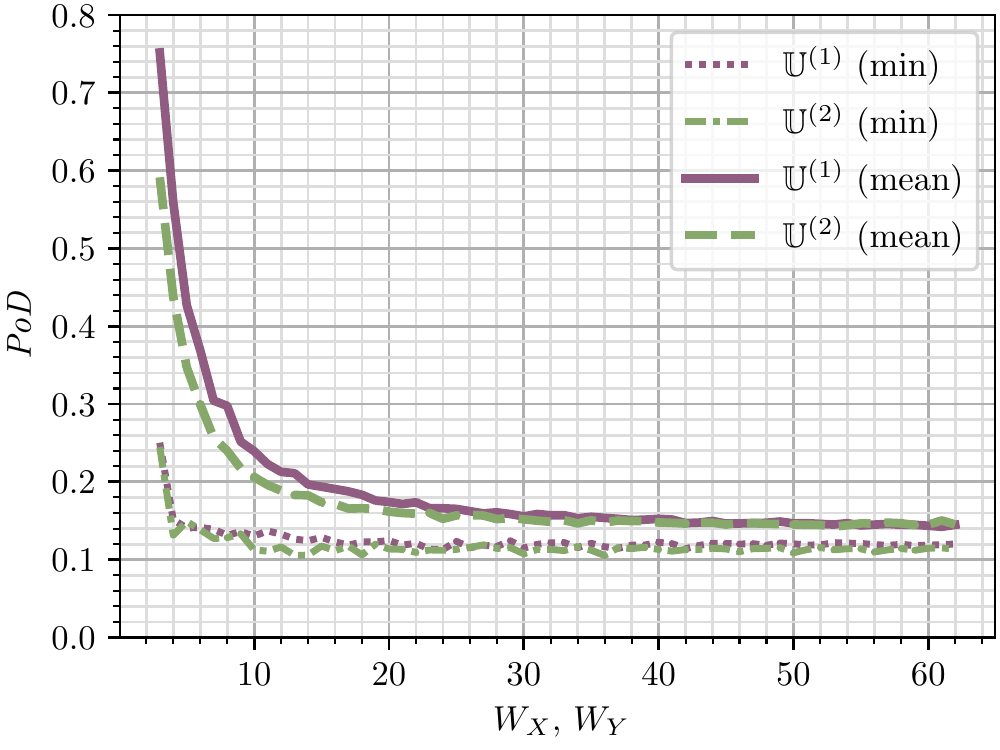}
    \caption{Price of Dishonesty (minimum and mean) guaranteed by BOSCO depending
    on the number of choices~$W_X = W_Y$
    for two different utility 
    distributions~$\mathbb{U}^{(1)}$ 
    and~$\mathbb{U}^{(2)}$.}
    \label{fig:negotiation:pod}
\end{figure}

It remains to analyze how the
choice sets should be constructed
such that equilibria with a 
low Price of Dishonesty result.
Surprisingly, we have found 
that random generation of such
choice sets works reasonably well
in practice. In particular, the
choice set~$V_X$ for any party~$X$ can be
constructed by sampling a high
enough number of choices~$v_X$
from the utility 
distribution~$\mathbb{U}_X$.
With multiple trials of such
random choice-set generation,
choice sets with a relatively
low Price of Dishonesty can be
found.

In~\cref{fig:negotiation:pod}, for
example, we analyze the 
resulting~$\mathit{PoD}$ from
random choice-set generation
for two uniform utility distributions,
namely $\mathbb{U}^{(1)}$, which
is a uniform distribution of~$(u_X,u_Y)$ on~$[-1, 1]\times[-1,1]$,
and $\mathbb{U}^{(2)}$, which is
a uniform distribution 
on~$[-\frac{1}{2},1]\times[-\frac{1}{2},1]$.
For each choice-set cardinality~$W_X = W_Y$,
we generate 200 random choice-set
combinations and find the mean and
the minimum of the 
associated~$\mathit{PoD}$ values.
Clearly, a higher number of choices
generally helps to reduce 
the Price of Dishonesty, 
but given around 50 choices,
adding more choices does not
improve the mechanism efficiency. Interestingly,
we also observe that the
number of \emph{equilibrium choices} (i.e., choices
which have a non-empty associated interval in the equilibrium
strategy) for each party
reaches 4 at that point
and is not further increased
for more possible choices.
Hence, the BOSCO service can
increase the number of possible choices
until the resulting~$\mathit{PoD}$
values do not substantially
decrease anymore. With this procedure,
the BOSCO mechanism could
guarantee a Price of Dishonesty
of around 10\% for 
both~$\mathbb{U}^{(1)}$ 
and~$\mathbb{U}^{(2)}$ in
the example at hand,
meaning that the negotiation
can be expected to be 10\%
less efficient than under
the unrealistic assumption
of perfect honesty.

\section{Effect on Path Diversity}
\label{sec:experiments}

In this section, we attempt to quantify
the effect of mutual\-ity-based
agreements on path diversity in the Internet.
Starting from the CAIDA
AS-relationship dataset~\cite{caida}, 
we construct a network of ASes
where a provider--customer or peering relationship results
in a single provider--customer or peering link, respectively.
In this graph, we generate
all possible mutuality-based agreements (MAs) for
the whole topology: For every pair $(A,B)$ of peers,
we generate an MA in which~$A$ gives
$B$ access to all its providers and peers 
which are not customers of $B$, and vice versa.
As MAs consist of an AS~$A$ 
giving  its peer $B$ access to a provider or 
another peer, these agreements enlarge the set of 
paths with 3 AS hops and
2 inter-AS links (henceforth: length-3 paths) for $B$, as well as the 
set of ASes that $B$ can reach with such length-3 paths (henceforth: nearby destinations). 

Given this graph and these MAs, we perform the following analysis
for 500 randomly chosen ASes. 
First, we 
find the GRC-con\-for\-ming length-3 paths
starting at the given AS. Then, 
we find the MA-created length-3 paths
for which the given AS is an end-point.
The number of these additional paths
and the number of additional nearby destinations
are thus metrics 
for the increased path diversity that the given AS enjoys
thanks to MAs.

\begin{figure}[tb]
    \centering
    \includegraphics[width=\columnwidth, trim=0 0 0 0]{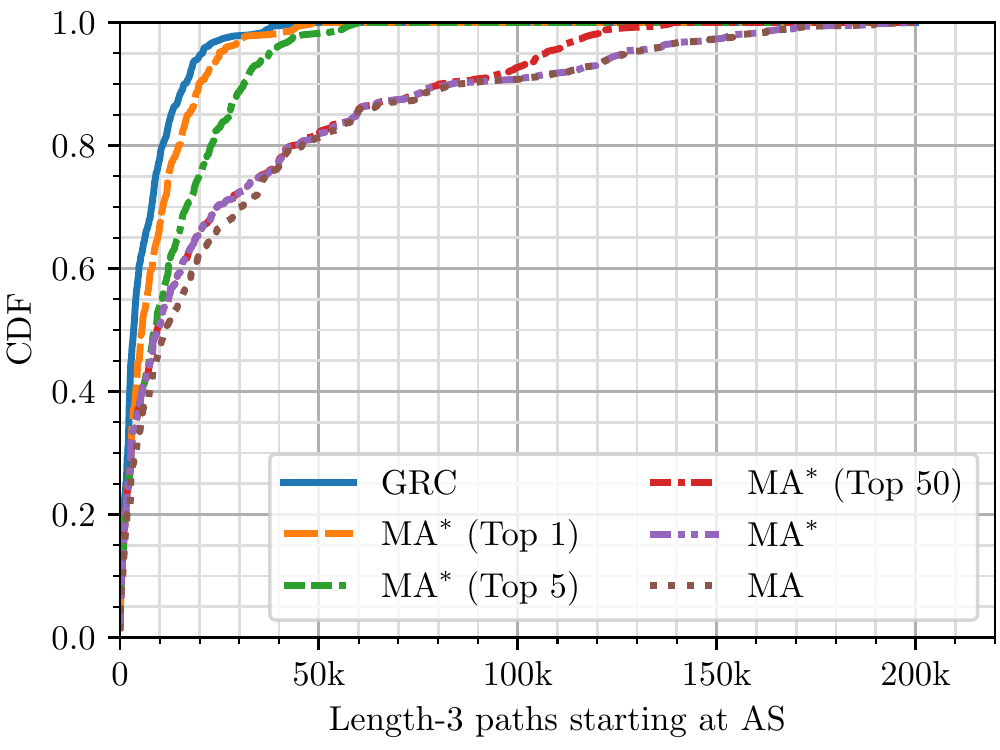}
    \caption{Distribution of ASes with respect
    to the number of length-3 paths starting
    at the AS, given different degrees of MA conclusion.}
    \label{fig:experiments:paths}
\end{figure}

\subsection{Number of Paths and
Nearby Destinations}
\Cref{fig:experiments:paths} shows the
substantial increase in the number of length-3 paths with
that are potentially available to these ASes thanks to
mutuality-based agreements: For example,
whereas none of the analyzed ASes have more
than 45,000 GRC-conforming paths with length~3,
20\% of the analyzed ASes have more than
45,000 length-3 paths if all 
MAs are concluded (CDF for 'MA'). Note that for a fixed
source and a fixed destination, all length-3 paths 
are disjoint by definition.

Since the conclusion of all
possible MAs is an extreme
case (although MAs could
be negotiated in an automated
fashion with the mechanism presented
in~\cref{sec:negotiation}), we further
analyze the effects of non-comprehensive
agreement conclusion. Initially, we note
that an MA can provide an AS with new paths
in two manners. First, an AS can \emph{directly}
gain an MA path by concluding an MA that includes
the path (e.g., as AS~$D$ gains the path~$DEB$
in~\cref{fig:model:new-examples} from
the MA with AS~$E$). Second, an AS
can \emph{indirectly} gain an MA path 
by being the subject of an MA 
that includes the
path (e.g., as AS~$B$ or AS~$F$ gain
paths to AS~$D$ from the MA 
between AS~$D$ and AS~$E$ 
in~\cref{fig:model:new-examples}).
Interestingly, most additional MA paths
are directly gained paths, 
as the similarity of the CDFs for
all MA paths~(MA) and directly gained
MA paths ($\mathrm{MA}^{\ast}$)
in~\cref{fig:experiments:paths} suggests.
Hence, the ASes bearing the negotiation
effort of an MA have a strong incentive to negotiate
that MA despite the effort,
because they typically are its biggest beneficiaries. 
Moreover, we find that an AS already tends to
obtain substantial gains in path diversity
with only a handful of MAs.
This point is demonstrated by the results
for the scenarios where an AS only concludes
the~$n$ MAs which provide it with the most
new paths 
(annotated with~`$\mathrm{MA}^{\ast}$ (Top~$n$)' 
in~\cref{fig:experiments:paths}):
Even if an AS only concludes the single most
attractive agreement from its perspective, 
it stands to gain several thousands 
of new paths.

\begin{figure}[tb]
    \centering
    \includegraphics[width=\columnwidth, trim=0 0 0 0]{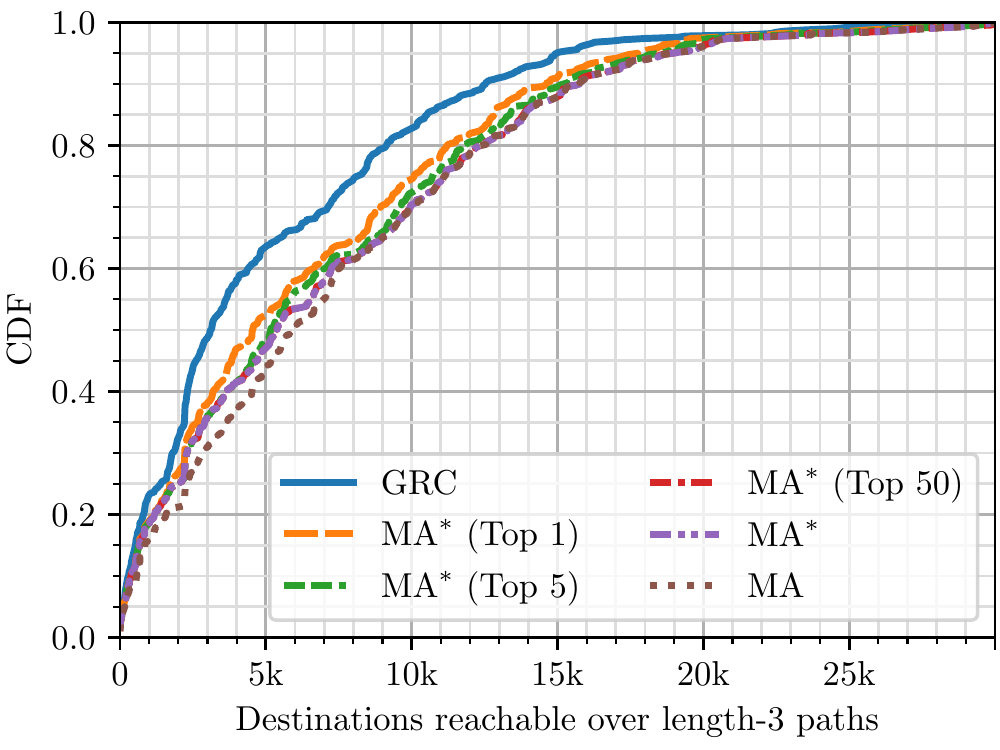}
    \caption{Distribution of ASes with respect
    to the number of destination reachable over length-3 paths, given different degrees of MA conclusion.}
    \label{fig:experiments:destinations}
\end{figure}

\Cref{fig:experiments:destinations} also
illustrates that mutuality-based agreements
enlarge the set of destinations reachable
with paths of length~3: For example, whereas
40\% of the analyzed ASes can reach more than
5,000 destinations over length-3 paths,
57\% of ASes can reach more than 5,000 destinations
over such paths if all MAs are concluded.
Interestingly, very few MAs per AS already suffice
to reap most of these benefits, as the results for
non-comprehensive agreement conclusion demonstrate.

For the set of analyzed ASes,
the average number of additional length-3
paths thanks to mutuality-based agreements is 
22,891 paths (maximum: 196,796 paths), 
and the average number of 
additionally reachable destinations over
length-3 paths is 2,181 ASes (maximum: 7,144 ASes).
Interestingly, the gains in terms of additionally
reachable destinations are more broadly distributed
than the gains regarding paths.
The explanation for this phenomenon is that
mutuality-based agreements in very densely connected
regions of the Internet lead to a high number
of additional length-3 paths, but have little impact on the number
of ASes reachable over such paths.

\subsection{Geodistance}
\label{sec:experiments:geodistance}

In order to gain a more qualitative perspective
on the additional paths enabled by MAs, we also
investigate the geographical length (henceforth:
geodistance) of these new paths. Such geodistance
is an important determinant of path latency~\cite{singla2014internet},
which is typically considered a core aspect
of path diversity. As the CAIDA AS-relationship
dataset~\cite{caida} does not directly contain
the necessary information, we additionally
build on the CAIDA prefix-to-AS dataset~\cite{caidaprefix}, the GeoLite2 database~\cite{geolite2}, and
the CAIDA geographic AS-relationship
dataset~\cite{caidageo}.
In particular, we determine the geolocation 
of any AS by finding the IP prefixes 
associated
with the AS number in the prefix-to-AS dataset,
determining the geolocation of the IP prefixes
via the GeoLite2 database, and averaging
the resulting coordinates to obtain
the center of gravity of the AS. 
With such averaging,
the potentially considerable intra-AS latency
of geographically distributed top-tier ASes
is automatically taken into account. 
Moreover, we obtain the geolocation
of an AS interconnection from the
CAIDA geographic AS-relationship dataset.
The geodistance of a length-3 
path~$\pi = (A_1, \ell_{12}, A_2, \ell_{23}, A_3)$, where~$A_i$ are ASes
and~$\ell_j$ are inter-AS links,
is then computed as~$d(\pi) = d(A_1,\ell_{12}) + d(\ell_{12},\ell_{23}) + d(\ell_{23}, A_3)$,
where~$d(X, Y)$ is the geodistance between
two points. If there are multiple known
AS interconnections, the geodistance
of the AS-level path~$(A_1,A_2,A_3)$ 
is computed for~$\ell_{12}$ and~$\ell_{23}$
that minimize~$d(\pi)$.

Using this measure of path geodistance, we again
compare the set of paths that conform to the GRC
and the set of paths enabled by novel MAs. For every analyzed AS pair connected by
at least one length-3 GRC path, we determine the maximum,
median, and minimum geodistance given
the length-3 GRC paths connecting the AS pair.
In a next step, we determine the geodistance of the additional
MA paths and check for each MA path whether it
is lower than the maximum, median, or minimum GRC geodistance,
respectively. Each AS pair is then characterized
by the number of MA paths below these
comparison thresholds. The aggregate results
of this comparison method are presented in~\cref{fig:experiments:distance}.

\Cref{fig:experiments:distance} shows that through MAs,
around 50\% of AS pairs gain at least 1 path with a lower
geodistance than the minimum-geodistance GRC path,
suggesting that inter-AS latency can be decreased
by MAs in these cases.
Around 25\% of AS pairs even gain at least 5 paths that improve
upon the minimum GRC geodistance, and at least 7 and 8 paths
that improve upon the median and maximum GRC geodistance,
respectively. Another 20\% of AS pairs only gain MA paths
with a higher geodistance than the maximum GRC geodistance
(or no new paths at all); however, also these additional 
paths have value in terms of reliability.
Regarding the AS pairs experiencing
a reduction in minimum 
geodistance,~\cref{fig:experiments:distance:reduction} illustrates the considerable extent
of the geodistance reduction for these
AS pairs: For example, 50\% of AS pairs
that experience a geodistance reduction
obtain a reduction of more than 24\%.

\begin{figure}[tb]
\begin{subfigure}[b]{\columnwidth}
    \includegraphics[width=\columnwidth, trim=0 0 0 0]{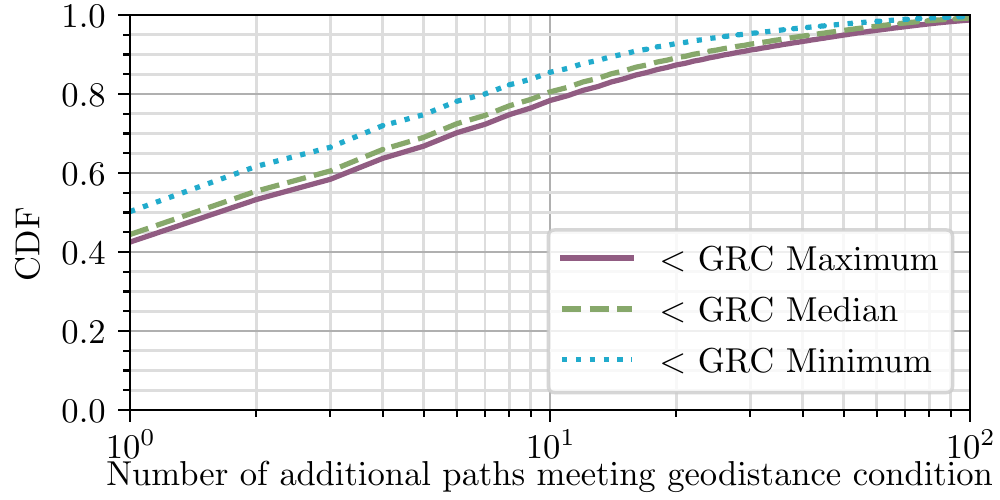}
    \caption{Distribution of AS pairs with respect to geographical length of additional paths from mutuality-based agreements, distinguished by
    satisfied geodistance-comparison
    condition.}
    \label{fig:experiments:distance}
\end{subfigure}

\begin{subfigure}[b]{\columnwidth}
    \centering
    \includegraphics[width=\columnwidth, trim=0 0 0 0]{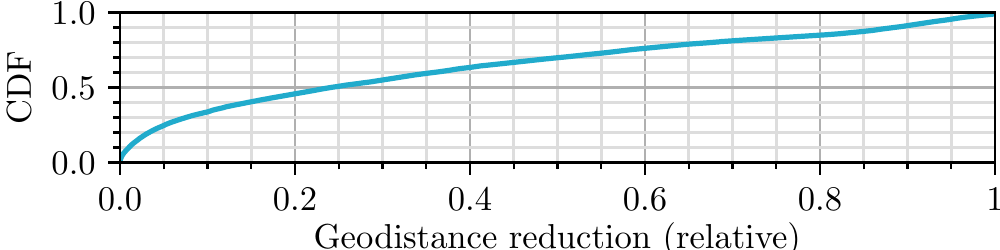}
    \caption{Distribution of relevant AS pairs (i.e., AS pairs which
    experience a geodistance reduction thanks
    to MAs) with respect to the relative geodistance reduction from mutuality-based agreements.}
    \label{fig:experiments:distance:reduction}
\end{subfigure}
\caption{Results of geodistance analysis.}
\label{fig:experiments:geodistance:total}
\end{figure}

\subsection{Bandwidth}
\label{sec:experiments:bandwidth}

We perform an analogous analysis
as in the preceding section with
respect to the bandwidth of
additional paths. To
infer the bandwidth of inter-AS
links, we employ a \emph{degree-gravity
model}~\cite{saino2013toolchain} 
which endows each link
with a capacity value proportional
to the product of the node degrees
of the link end-points. The path
bandwidth is then
the minimum such computed link bandwidth 
of all links in the path.

With such an analysis, we find that
35\% of all investigated AS pairs
obtain a new MA path that has more
bandwidth than the corresponding 
maximum-bandwidth GRC path 
(cf.~\cref{fig:experiments:bandwidth}).
Of these benefiting AS pairs,
50\% gain an MA path with
at least 150\% more bandwidth than
the respective maximum-bandwidth
GRC path (cf.~\cref{fig:experiments:bandwidth:increase}).

\begin{figure}[t]
\begin{subfigure}[tb]{\columnwidth}
    \centering
    \includegraphics[width=\columnwidth, trim=0 0 0 0]{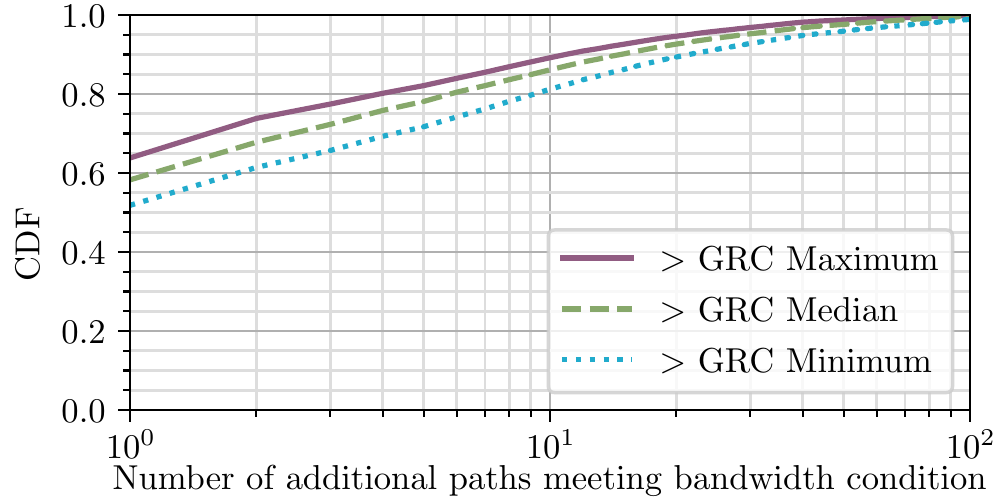}
    \caption{Distribution of AS pairs with respect to bandwidth of additional paths from mutuality-based agreements, distinguished
    by satisfied bandwidth-comparison condition.}
    \label{fig:experiments:bandwidth}
\end{subfigure}

\begin{subfigure}[b]{\columnwidth}
    \centering
    \includegraphics[width=\columnwidth, trim=0 0 0 0]{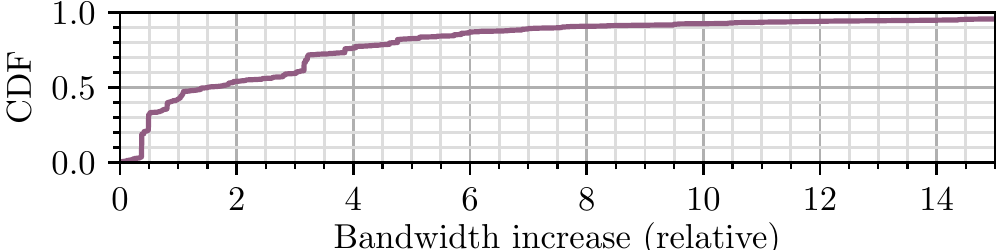}
    \caption{Distribution of relevant AS pairs (i.e., AS pairs which
    experience a bandwidth increase thanks
    to MAs) with respect to the relative bandwidth increase from mutuality-based agreements.}
    \label{fig:experiments:bandwidth:increase}
\end{subfigure}
\caption{Results of bandwidth analysis.}
\label{fig:experiments:bandwidth:total}
\end{figure}
\section{Related Work}
\label{sec:related-work}

After the growth of the Internet
had led to considerable BGP stability problems
in the late 1990s~\cite{govindan1997analysis,labovitz1999experimental}, the research of AS
interconnection agreements, their stability
properties, and their optimal structures
received significant interest.
The commercial reality of the Internet
has been shown to mainly contain two basic types of
agreements that determine
route-forwarding policies, namely pro\-vider--customer
agreements and peering
agreements~\cite{huston1999interconnection,norton2011internet}. The relative exclusiveness of these two
agreement types was reinforced by the important
result of Gao and Rexford, showing that BGP route
convergence is guaranteed if ASes stick to these
two forms of route-forwarding
policies~\cite{gao2001stable}.

However, it is well-known that such strict
BGP policies reduce path quality:
For a majority of routes selected in BGP,
there exists a route that is more
attractive with respect to metrics such as bandwidth,
latency, or loss
rate~\cite{savage1999end,kotronis2016stitching,gupta2015sdx,giotsas2013inferring,gill2013survey,qiu2007toward}.
In fact, motivated by such improvements, already today, many ASes do
not always follow the Gao--Rexford
conditions~\cite{gill2013survey}.
For example, some ISPs
use alternative paths to reach content distribution networks such as
Akamai~\cite{anwar2015investigating,mazloum2014violation,giotsas2014inferring_complex},
other ISPs prefer the peer route through their Tier-1 neighbor
over a longer customer route~\cite{anwar2015investigating},
and so on.
Still, these deviations are narrow in scope, with most non-GRC
policies being explainable by ``sibling'' ASes, which belong to the
same organization and provide mutual transit
services~\cite{gao2001inferring,anwar2015investigating}, and
partial/hybrid provider relationships~\cite{giotsas2014inferring_complex}.
This is due to the fact that more complex policies would threaten
the convergence of the routing process unless they are supported through
multi-AS coordination efforts. These restrictions have also been
acknowledged by previous efforts to provide multipath routing
based on BGP such as MIRO~\cite{xu2006miro}.

While interconnection agreements in \pan{} architectures~%
\cite{Raghavan2004,Raghavan2009,Bhattacharjee2006,Yang2007,
godfrey2009pathlet,Anderson2013,Zhang2007,perrig2017scion,
trammell2018adding} do not need to follow
the guidelines devised to achieve BGP
stability, these agreements should definitely
also respect the economic self-interest of ASes.
The Gao--Rexford guidelines for BGP policies
have been proven to be rational in that
sense~\cite{feigenbaum2006incentive}.
Notable proposals for agreement structures
that attempt to satisfy both AS self-interest
and global efficiency
include Nash peering~\cite{dhamdhere2010value, zarchy2018nash}, where the cooperative
surplus of the agreement is shared among
the parties according to the Nash bargaining solution~\cite{nash1950bargaining}, and
ISP-settlement mechanisms based on
the Shapley value~\cite{ma2010internet}.
It is important to note that
unlike traditional source routing and
similar to MIRO~\cite{xu2006miro},
\pan{s} still offer transit ASes control over
the traffic traversing their network, and hence to
maximize their revenue.
In contrast to MIRO, however, PANs guarantee
path stability.



Finally, our paper has parallels to the work
by Haxell and Wilfong~\cite{10.5555/1347082.1347104},
who showed that a fractional relaxation of the
stable-paths problem of BGP~\cite{griffin2002stable}
guarantees a solution
and that a more flexible routing paradigm
can thus defuse BGP stability issues.
In contrast to this work, however,
our focus lies on the interconnection agreements
under this new paradigm, the extent to which
these agreements increase path diversity,
and their economic rationality
and bargaining aspects.


\section{Conclusion}
\label{sec:conclusion}

This work shows that \pan{} architectures 
enable novel types of interconnection agreements,
thereby substantially improving path diversity
in the Internet and creating new business 
opportunities.
Such new possibilities
exist in \pan{} architectures as they do not rely on the
nowadays essential route-forwarding policy
guidelines formulated by
Gao and Rexford~\cite{gao2001stable} for
route convergence.

Our results show
that path diversity in the Internet benefits
enormously by enabling paths beyond the Gao--Rexford
constraints: By using previously impossible path types,
an AS can reach thousands of new destinations
with 3-hop paths and benefit from hundreds
of thousands of additional paths, some of
which have more desirable characteristics
than the previously available paths. 
There is thus a largely unknown advantage to \pan{}
 architectures: Not only do these
architectures enable end-hosts to \emph{select}
a forwarding path, they also
allow network operators to \emph{offer}
new (and often shorter) forwarding paths.
As \pan{s} are not limited to using a single
path between a pair of ASes, all these paths can 
be used simultaneously according to the requirements
of end-hosts and their applications (e.g.,
low latency for voice over IP and high bandwidth for
file transfers). These direct benefits to end-hosts 
in turn incentivize providers to explore new
interconnection agreements and offer diverse paths
to attract new customers.

We present two methods for designing
agreements that are Pareto-optimal, fair, and 
thus attractive to both parties. We also
show that, assisted by an appropriate
bargaining mechanism, the negotiation of such
agreements can lead to efficient agreements
although necessary information is private. 

We see this work merely as a first step in
exploring the new possibilities for interconnection
agreements in \pan{} architectures. 
There are many exciting opportunities 
for future research in designing and evaluating 
interconnection agreements that can achieve
desirable goals of network operators, such as
network utilization, predictability, and 
security.

\section{Acknowledgments}

We gratefully acknowledge support from ETH Zurich, from the Zurich
Information Security and Privacy Center (ZISC), from SNSF for
project ESCALATE (200021L\_182005) and from WWTF for 
project WHATIF (ICT19-045, 2020-2024). 
Moreover, we thank Giacomo Giuliari and Joel Wanner
for helpful discussions that improved this 
research.

\bibliographystyle{abbrv} 
\begin{small}
\bibliography{refs}
\end{small}
\vspace*{0mm}

\end{document}